\documentclass[11pt,twoside]{article}

\usepackage[ansinew]{inputenc} 
\usepackage{amsmath,amsfonts,amssymb,amsthm}
\usepackage{pstricks,pst-node,pst-coil,pst-plot,pstricks-add}
\usepackage{graphics,geometry,epsfig}
\usepackage{bm,bbm}

\usepackage{mathrsfs} 

\usepackage[pdftex]{hyperref}

\usepackage{graphicx}

\usepackage{units}
\usepackage[squaren]{SIunits}

\usepackage{tabularx,dcolumn,booktabs}

\usepackage{rotate}

\usepackage{float}

\usepackage{dsfont}

\usepackage{chngcntr}

\usepackage{icomma}

\usepackage{dcolumn}
\usepackage{multirow}
\usepackage{booktabs}

\usepackage[width = 0.9\linewidth, textfont = {small}, labelfont = {small, bf}, format = plain]{caption}

\usepackage[innercaption]{sidecap}

\newcommand{\field}[1]{\mathbb{#1}}
\newcommand{\N}{\field{N}}
\newcommand{\R}{\field{R}}
\newcommand{\C}{\field{C}}

\newcommand{\HH}{\mathscr H}

\newcommand{\DD}{\mathcal D}

\newcommand{\FF}{\mathcal F}

\newcommand{\hh}{\mathfrak h}

\newcommand{\ID}{\mathds{1}} 

\newcommand{\Ima}{\operatorname{Im}}
\newcommand{\Rea}{\operatorname{Re}}

\newcommand{\limop}[2][\infty]{\lim\limits_{#2 \rightarrow #1}}

\newcommand{\liminfop}[2][\infty]{\liminf\limits_{#2 \rightarrow #1}}
\newcommand{\limsupop}[2][\infty]{\limsup\limits_{#2 \rightarrow #1}}

\newcommand{\f}{v}
\newcommand{\chibar}{\overline{\chi}}
\newcommand{\eps}{\varepsilon}

\newcommand{\abs}[1]{\mbox{$\left| #1 \right| $}}                 
\newcommand{\norm}[1]{\mbox{$\left\| #1 \right\|$}}           
\newcommand{\sprod}[2]{\mbox{$\left\langle #1,#2 \right\rangle$}}        
\newcommand{\fock}{\FF}
\newcommand{\hilbert}{\HH}
\newcommand{\finitepart}{\hilbert_0}
\newcommand{\cutoff}{\Lambda}
\newcommand{\ir}{K}         

\newtheorem{theorem}{Theorem}[section]
\newtheorem{lemma}[theorem]{Lemma}
\newtheorem{corollary}[theorem]{Corollary}

\newtheorem{prop}[theorem]{Proposition}

\theoremstyle{plain}

\title{Self-Adjointness and Domain of the Fr\"ohlich Hamiltonian}

\author{M.~Griesemer\footnote{marcel.griesemer@mathematik.uni-stuttgart.de} and A.~W\"unsch\footnote{andreas.wuensch@mathematik.uni-stuttgart.de}\\  
\small Fachbereich Mathematik, Universit\"at Stuttgart, D-70569 Stuttgart, Germany}  

\date{}

\begin{document}
\maketitle

\begin{abstract}
In the large polaron model of H.~Fr\"ohlich, the electron-phonon
interaction is a small perturbation in form sense, but a large
perturbation in operator sense. This means that the form-domain of
the Hamiltonian is not affected by the interaction but the domain of
self-adjointness is.  In the particular case of the Fr\"ohlich model, we are
nevertheless able, thanks to a recently published new operator bound, to give an explicit characterization of the domain
in terms of a suitable dressing transform. Using the mapping properties of
this dressing transform, we analyse the smoothness of vectors in the
domain of the Hamiltonian with respect to the position of the electron. 
\end{abstract}


\section{Introduction}   \label{sec:introduction}

A popular model for the description of an electron in a polar crystal due to Fr\"ohlich, Pelzer and Zienau is based on the formal expression 
\begin{equation}
\label{def-H}
     -\Delta +N + \sqrt{\alpha}\int_{\R^3}[e^{ikx} a(k) + e^{-ikx}a^{*}(k)]\,\frac{dk}{\abs{k}}
\end{equation}
for the Hamiltonian of the system~\cite{FPZ1950}. Here,
$\Delta$ denotes the Laplace operator in $L^2(\R^3)$, $N$ is the number
operator in the symmetric Fock space over $L^2(\R^3)$, and $\alpha$ is a
coupling constant. The third term of \eqref{def-H}, which accounts for
the electron-phonon interaction, is not an operator in the Hilbert
space because the form factor is not square integrable.  Therefore,
expression \eqref{def-H}, as it stands, is not a densely defined
operator and hence cannot readily be adopted as the Hamiltonian of the
system. Expression \eqref{def-H} does, however, define a closed,
semi-bounded quadratic form with domain $D(H_0^{1/2})$, where $H_0=-\Delta+N$. Indeed, by a
simple argument of Lieb and Thomas, the interaction is
infinitesimally form bounded with respect to $H_0$~\cite{LiebYamazaki1958,LiebThomas1997}. There is
therefore a unique self-adjoint operator $H$, the
Fr\"ohlich Hamiltonian, associated with the quadratic form
defined by~\eqref{def-H}. If $H_{\cutoff}$, for $\cutoff>0$, is defined in
terms of \eqref{def-H} with ultraviolet cutoff
$\abs{k}\leq \cutoff$ in the interaction, then it follows, by general arguments, that $H_{\cutoff}\rightarrow H$ in the norm resolvent sense as $\cutoff\rightarrow \infty$.

The main purpose of this paper is to describe the domain $D(H)$ of
$H$ as explicitly as possible.  To this end, we follow Nelson and determine
$UHU^{*}$, where $U$ is a dressing transform given by Gross \cite{Nelson1964}.  Using a
recently published new variant of the Lieb-Thomas bound, we are able to show that $UHU^{*}$ is self-adjoint on
$D(H_0)$ and hence that
\begin{equation}\label{domain-H} 
     D(H) = U^* D(H_0).
\end{equation}
This result allows us to determine a core of $H$ in terms of coherent
states and to describe the action of $H$ on this core
explicitly. Moreover, we show that 
\begin{equation}\label{D-supset} 
     D(H) \subset \Big(\bigcap_{0<s<3/4} D\big((-\Delta)^{s}\big)\Big)\cap D(N),
\end{equation}
and that 
\begin{equation}\label{D-cap} 
     D(H) \cap D\left((-\Delta)^{3/4}\right) = \{0\}.
\end{equation}
The identity \eqref{D-cap} implies in particular that $D(H) \cap D(-\Delta) = \{0\}$, which
has the following simple explanation: when $H$ is applied to a
vector $\Psi \in D(H)\backslash\{0\}$, then the interaction part in
\eqref{def-H}, we call it $\sqrt{\alpha}W$, creates a vector
$\sqrt{\alpha}W\Psi$ outside of the Hilbert space. In fact,
$\sqrt{\alpha}W\Psi$ belongs to the dual of $D(H_0^{1/2})$ equipped
with the form norm of $H_0$. This vector must be canceled by some part of $H_0\Psi$ that is not in the
Hilbert space either. This means that $\Psi\not\in D(H_0)$ and, since
$\Psi\in D(N)$, by \eqref{D-supset}, we conclude that $\Psi\not\in
D(-\Delta)$. The mechanism of this cancellation of non-Hilbert space
parts is illustrated in the appendix by a formal computation of \eqref{def-H}
applied to vectors $\Psi$ from a core of $H$ where we know the action
of $H$ explicitly. Of course, these remarks equally apply to other Hamiltonians
describing quantum particles interacting with a quantized field of
bosons. Indeed, we prove \eqref{domain-H} and suitable generalizations
of \eqref{D-supset} and \eqref{D-cap} for a large class of form factors $v(k)$ including
$v(k)=|k|^{-(d-1)/2}$, $k\in \R^d$, describing the polaron in $d=2$ and $d=3$
space dimensions, respectively. In this more general case, the
admissible exponents in \eqref{D-supset} and \eqref{D-cap}
are determined by the rate of decay of the form factor as
$|k|\to\infty$. Our results could be further generalized to include  
$N$-polaron systems or external magnetic fields, but we refrain from such generalizations in order to keep the
paper short and the notation simple.

For the massive Nelson model where
$H_0=-\Delta+d\Gamma(\omega)$, $\omega(k)=\sqrt{k^2+m^2}$ and
$v(k)=\omega(k)^{-1/2}$, we expect results analogue to
\eqref{D-supset} and \eqref{D-cap}. In that case, however, the role of the number
operator $N$ is played by the field energy $d\Gamma(\omega)$. Its
domain is not left invariant by the Gross transform, which complicates
matters. We plan to return to this case in a future publication.

The UV renormalization of the Nelson and the Fr\"ohlich models in terms
of the Gross transform is well-understood and well-documented in the
literature~\cite{Nelson1964,Froehlich1974,Loewen1988,AmmariZied2000}. The
much more direct and straightforward characterizations of $H$ based on
the Lieb-Thomas argument have not yet been properly described in the literature, and we
therefore elaborate on them in Section \ref{sec:FroehlichHamiltonian}. Our main objectives are, however, characterization \eqref{domain-H} of the domain,
see Section~\ref{sec:GrossTransformation}, and the proofs of \eqref{D-supset} and
\eqref{D-cap} in Section~\ref{sec:regularity}. In the appendices, we
prove an abstract result on resolvent convergence based on form
bounds, Appendix~\ref{app:abstract}, we collect background on annihilation and creation
operators, Appendix~\ref{app:creation}, and we describe the action of $H$ on vectors
from a suitable core of $H$, Appendix~\ref{app:formal}.

\section{The construction of the Fr{\"o}hlich Hamiltonian} 
\label{sec:FroehlichHamiltonian}

In this section, we describe the class of Hamiltonians whose domains
will be studied in Sections~\ref{sec:GrossTransformation} and~\ref{sec:regularity}. These
Hamiltonians describe a quantum particle (called electron) in $\R^d$ that is coupled linearly to a quantized field of scalar bosons (called phonons). We begin with notations and hypotheses on the form factors. 

Let $\hilbert:=L^2(\R^d,dx)\otimes\fock$, where $\fock$ denotes the symmetric Fock space over $L^2(\R^d,dk)$ with arbitrary $d\in\N$. We may identify $\hilbert$ with $L^2(\R^d,\fock)$ through the isomorphism given by $\varphi\otimes\eta\mapsto\varphi(x)\eta$. Let $H_0:=-\Delta+N$, where $\Delta$ is the (self-adjoint) Laplacian in $L^2(\R^d)$ and $N$ denotes the number operator in $\fock$. Let
$$
   \norm{\Psi}_0:=\norm{(H_0+1)^{1/2}\Psi}
$$
for $\Psi\in D(H_0^{1/2})$. The Hamiltonian $H_0$ is self-adjoint on $D(H_0)=D(-\Delta\otimes\ID)\cap D(\ID\otimes N)$ and essentially self-adjoint on $D(H_0)\cap\finitepart$, where 
$$
   \finitepart:=\bigcup\limits_{n\geq 0}\chi(N\leq n)\hilbert.
$$

The electron-phonon interaction occurs in terms of annihilation and
creation of phonons. The usual annihilation and creation operators in Fock space associated
with some vector $f\in L^2(\R^d)$ will be denoted by $a(f)$ and
$a^{*}(f)$, respectively. They are closed, adjoint to each
other with $D(a(f)) = D(a^{*}(f)) \supset D(\sqrt{N})$, and they obey
the canonical commutation relations $[a(f),a^{*}(g)] =
\sprod{f}{g}$ on $D(N)$. The symmetric field operators 
$$
     \phi(f) := a(f) + a^{*}(f),\qquad \pi(f) := \phi(if)
$$
are essentially self-adjoint on $D(N)$, and they obey the commutation
relations 
\begin{align*}
     [\phi(f) ,\phi(g)] &= 2i\Ima \sprod{f}{g},\\
      [\phi(f) ,\pi(g)] &= 2i\Rea \sprod{f}{g}.
\end{align*}
The (self-adjoint) closures of the operators $\phi(f)$ and $\pi(f)$ will be denoted by the same symbols.

We will have occasion to work with generalized annihilation and
creation operators $a(F)$ and $a^{*}(F)$ that are operators in
$\hilbert$ rather than $\fock$. Here
$F:L^2(\R^d,dx)\to L^2(\R^d,dx)\otimes L^2(\R^d,dk)$ is a linear operator. In the simplest
case, $F\varphi=\varphi\otimes f$ for some $f\in L^2(\R^d,dk)$ and then
$a^{\#}(F)=\ID\otimes a^{\#}(f)$ is the usual annihilation or creation operator in
$\fock$. Often, but not always, the operator $F$ will be defined in terms
of some function $(x,k)\mapsto F_x(k)$, denoted by $F$ as well,
through the equation $(F\varphi)(x,k)=\varphi(x)F_x(k)$. In this case, $(a^{\#}(F)\Psi)(x)
=a^{\#}(F_x)\Psi(x)$. Typically, $F_x(k)=e^{-ikx}f(k)$, where $f \in
L^2(\R^d)$ and then the operator norm of $F$ equals the norm
of $f$ in $L^2(\R^d)$. See Appendix~\ref{app:creation} for the definition
of $a^{\#}(F)$ in the general case.

For $\cutoff<\infty$, we define $H_{\cutoff}:D(H_0)\subset\hilbert\rightarrow\hilbert$ by
$$
   H_{\cutoff}:=H_0+\phi(G_{\cutoff}),
$$
where
$$
   G_{\cutoff,x}(k):=e^{-ikx}\f(k)\chi_{\cutoff}(k).
$$
Here, $\chi_{\cutoff}$ denotes the characteristic function of the set
$\{k\in\R^d|\ \abs{k}\leq\cutoff\}$. On the form factor $\f:\R^d\rightarrow\C$, we impose the following assumptions:
\begin{align*}
   (\f 1)\ \ \quad &\f\in L_{loc}^2(\R^d)\quad\text{and}\quad\f(k)=\f(-k), \\
   (\f 2)\ \ \quad &\int\frac{\abs{\f(k)}^2}{1+k^2}\,dk<\infty.
\end{align*}
These assumptions are sufficient for the results of the present
section. Later, we will replace $(\f 2)$ by the slightly stronger assumption
\begin{align*}
   (\f 3)\ \ \quad \sup\limits_{q\in\R^d}\int\limits_{\abs{k}\geq
     \ir}\frac{\abs{\f(k)}^2}{1+(q-k)^2}dk \longrightarrow 0\qquad
   (\ir \to\infty). 
\end{align*}
An example of a form-factor $\f$ satisfying these conditions is the function
\begin{align}
\label{v-example}
   \f(k)=\abs{k}^{-(d-1)/2},\quad d\geq 2, 
\end{align}
which includes the form factors $\f(k)=\abs{k}^{-1/2}$ and $\f(k)=\abs{k}^{-1}$ of the large polaron models in $d=2$ and $d=3$ dimensions, respectively. Use H{\"o}lder's inequality with  exponents $(d+1)/(d-1)$ and $(d+1)/2$ to see that~\eqref{v-example} satisfies $(\f 3)$. 

From Lemma~\ref{lm:a-like-root-N}, it follows that
$\phi(G_{\cutoff})(N+1)^{-1/2}$ is bounded, and hence, $\phi(G_{\cutoff})$ is infinitesimally $H_0$-bounded. We conclude, by Kato-Rellich, that $H_{\cutoff}$ is self-adjoint on $D(H_0)$ and, moreover, that the quadratic form
\begin{align}
\label{equ:quadraticForm}
   (\Phi,\Psi) \mapsto \sprod{\Phi}{\phi(G_{\cutoff})\Psi}
\end{align}
defined on $D(H_0^{1/2})$ satisfies the Hypothesis~(a) of Theorem~\ref{thm:Ammari}. Let
\begin{equation*}
       D_{\ir}:=\int_{\abs{k}\geq\ir} \frac{1}{k^2} \abs{\f(k)}^2 dk.
\end{equation*}
Then, $D_{\ir}\rightarrow 0$ as $\ir\rightarrow\infty$ by Assumption~$(\f 2)$. Therefore, the following lemma establishes the Hypothesis~(b) of Theorem~\ref{thm:Ammari}:

\begin{lemma}
\label{lm:dW}
Assume $(\f 1)$ and $(\f 2)$. Then, for all $\cutoff_1,\cutoff_2\in\R_+$ and all $\Phi,\Psi\in D(H_0^{1/2})$, we have 
\begin{align*}
  \abs{\sprod{\Phi}{\phi(G_{\cutoff_1})\Psi}-\sprod{\Phi}{\phi(G_{\cutoff_2})\Psi}} \leq\abs{D_{\cutoff_1} -  D_{\cutoff_2}}^{1/2}\norm{\Phi}_0 \norm{\Psi}_0.
\end{align*}
\end{lemma}

\begin{proof}
Since $\phi(G_\cutoff)$ is symmetric, it suffices to
establish the desired bound for the case $\Phi=\Psi$. To this end, fix
$\Lambda_1,\Lambda_2>0$ and let 
\begin{align*}
   A_{x}(k) := \frac{k}{\abs{k}^2}  (G_{\cutoff_1,x}(k) - G_{\cutoff_2,x}(k))
\end{align*}
with components $ A_{x,\ell}(k)$, $\ell=1\ldots d$. Then,
$i\nabla_x \cdot A_{x}= G_{\cutoff_1,x}-G_{\cutoff_2,x}$, and hence, 
\begin{align*}
    [p, a(A)] := \sum_{\ell=1}^d
    [p_\ell,a(A_{\ell})]= a(G_{\cutoff_1} - G_{\cutoff_2}),
\end{align*}
where $p=-i\nabla$. It follows that
\begin{eqnarray*}
    \abs{\sprod{\Psi}{\phi(G_{\cutoff_1})\Psi}-\sprod{\Psi}{\phi(G_{\cutoff_2})\Psi}}
    &=&\abs{2\Rea\sprod{\Psi}{a(G_{\cutoff_1}-G_{\cutoff_2})\Psi}}\\
                                                 &=&\abs{2\Rea\sprod{\Psi}{[p,a(A)]\Psi}}\\
                                                &\leq &2\Big( \abs{\sprod{p\Psi}{a(A)\Psi}}\ + \abs{\sprod{a^{*}(A)\Psi}{p\Psi}}\Big)\\
                                                &\leq &4\norm{p\Psi} \norm{\sqrt{N+1}\Psi} \norm{A}\\
                                                &\leq &2\norm{\Psi}_0^2\norm{A},
\end{eqnarray*}
where $\norm{A}^2 = \abs{D_{\cutoff_1} -  D_{\cutoff_2}}$.
\end{proof}
We are now ready to prove the main result of this section.

\begin{theorem}
\label{thm:nrs}
Assume $(\f 1)$ and~$(\f 2)$. Then, the following statements hold true:
\begin{itemize}
\item[(i)]   The limit
  $W_{\infty}(\Phi,\Psi):=\limop{\cutoff}\sprod{\Phi}{\phi(G_{\cutoff})\Psi}$
  exists for all $\Phi,\Psi\in D(H_0^{1/2})$.
\item[(ii)]  The quadratic form on $D(H_0^{1/2})$ given by $\sprod{H_0^{1/2}\Phi}{H_0^{1/2}\Psi}+W_{\infty}(\Phi,\Psi)$ is closed and bounded from below.
\item[(iii)] If $H$ denotes the (unique) self-adjoint operator associated with the quadratic form from~(ii), then $H_\cutoff\rightarrow H$ in the norm-resolvent sense as $\cutoff\rightarrow\infty$.
\end{itemize}
\end{theorem}

\begin{proof}
We apply Theorem~\ref{thm:Ammari} to the Hamiltonian $H_0$ and the quadratic form defined by~\eqref{equ:quadraticForm}. We already pointed out that Hypothesis~(a) of Theorem~\ref{thm:Ammari} is satisfied and Hypothesis~(b) follows from Lemma~\ref{lm:dW}. Now the statements (i), (ii), and (iii) follow from Theorem~\ref{thm:Ammari}.
\end{proof}

The convergence $H_\cutoff\rightarrow H$ in the norm-resolvent sense implies convergence in the strong resolvent sense, which is equivalent to 
\begin{equation}
\label{eq:dyn}
      e^{-iH_\cutoff t}\Psi \longrightarrow e^{-iH t}\Psi\qquad (\cutoff\to\infty)
\end{equation}
for all $t\in \R$ and all $\Psi\in \HH$. Alternatively, the existence of limit
\eqref{eq:dyn} can be derived directly from Lemma~\ref{lm:dW} and its
Corollary~\ref{lm:W-is-small}, below. Hence, with the help of Stone's theorem, 
a further and very straightforward characterization of $H$ as the generator of
unitary group \eqref{eq:dyn} is achieved. This is the content of the
Theorem~\ref{thm:U-exists} and its proof. 

\begin{corollary}\label{lm:W-is-small}
Assume $(\f 1)$ and $(\f 2)$. Then, for every $\eps>0$, there exists $C_\eps$ such that for all $\cutoff>0$,
\begin{itemize}
\item[(a)] $\pm \phi(G_{\cutoff})\leq \eps H_0 +C_{\eps}$,
\item[(b)] $(1-\eps) H_0 - C_\eps \leq H_{\cutoff} \leq  (1+\eps)H_0 +
  C_\eps$.
\end{itemize}
\end{corollary}

\begin{proof}
(b) follows immediately from (a). To prove (a), note that the asserted
inequality is true for any fixed $\cutoff=\cutoff_0$. Then, choose
$\cutoff_0$ sufficiently large and use Lemma~\ref{lm:dW}.
\end{proof}

\begin{theorem}
\label{thm:U-exists}
Assume $(\f 1)$ and $(\f 2)$. Then, for all $t\in\R$ and $\Psi\in \HH$, the limit 
\begin{align*}
    U(t)\Psi := \limop{\cutoff} e^{-iH_\cutoff t}\Psi
\end{align*}
exists and defines a strongly continuous unitary group $U(t)$.
\end{theorem}

\begin{proof}
Let $U_\cutoff(t) = \exp(-iH_{\cutoff} t)$. Then, for all $\Psi\in D(H_0)$,
\begin{eqnarray}
   \norm{U_{\cutoff_1}(t)\Psi - U_{\cutoff_2}(t)\Psi}^2&=& 2\norm{\Psi}^2 - 2\Rea\sprod{U_{\cutoff_1}(t)\Psi}{U_{\cutoff_2}(t)\Psi} \nonumber \\
                                                 &=& -2\Rea \sprod{\Psi}{(U_{\cutoff_1}^{*}(t)U_{\cutoff_2}(t)-1)\Psi} \nonumber \\
                                                 &=& -2\Rea i \int_0^t \sprod{\Psi}{U_{\cutoff_1}^{*}(s)(H_{\cutoff_1}-H_{\cutoff_2} ) U_{\cutoff_2}(s)\Psi}\, ds \nonumber \\
                 \label{eq:dU}               &\leq & C \abs{t} \ \abs{D_{\cutoff_2}-D_{\cutoff_1}}^{1/2} \norm{\Psi}_0^2.
\end{eqnarray}
In the last inequality, we used Lemma~\ref{lm:dW} and Corollary~\ref{lm:W-is-small} (b), which implies that 
\begin{align*}
      \norm{U_\cutoff(t)\Psi}_0 \leq C\norm{\Psi}_0
\end{align*}
with a constant $C$ that is independent of $\cutoff$ and $t$. The bound \eqref{eq:dU} implies that $U(t)\Psi$ exists for all  $\Psi\in D(H_0)$ and that $U_{\cutoff}(t)\Psi \rightarrow U(t)\Psi$ uniformly for $t$ from compact intervals. Hence, $t\mapsto U(t)\Psi$ is continuous for $\Psi\in D(H_0)$. Since $\norm{U_{\cutoff}(t)}=1$ and since $D(H_0)$ is dense, it follows that $U(t)$ exists on $\HH$, that $\norm{U(t)}=1$, and that $t\mapsto U(t)\Psi$ is continuous. The group properties 
\begin{align*}
    U(0)=1\quad \text{and}\quad U(t+s) = U(t)U(s)
\end{align*}
follow from the corresponding properties of $U_\cutoff(t)$. They imply that $U(-t) = U(t)^{-1}$ and hence that $U(t)$ is unitary.
\end{proof}

\section{The Gross transform and the domain of $H$} 
\label{sec:GrossTransformation}

In this section, we prove Equation~\eqref{domain-H} in the Introduction
in the more general form given in Theorem~\ref{thm:domain-H},
below. To this end, we first need to recall, from \cite{Nelson1964},
the dressing transform of Gross and its effect on $H_{\cutoff}$.

The Gross transform $U_{\cutoff}:\hilbert\to\hilbert$ is a unitary linear operator depending on
the parameters $\ir,\cutoff\geq 0$, where $\ir$ is fixed
most of the time and, therefore, often suppressed in our notation. For given $\ir,\cutoff$ with $0<\ir<\cutoff\leq \infty$, we define
\begin{equation*}
   U_{\cutoff}:=e^{i\pi(B_{\cutoff})},
\end{equation*}
where
\begin{equation}\label{def:B-cutoff}
   B_{\cutoff,x}(k) :=-\frac{1}{1+k^2}G_{\cutoff,x}(k)(1-\chi_{\ir}(k)).
\end{equation}
We will use $kB_\Lambda$ and $k^2B_\Lambda$ to denote the functions
$kB_{\Lambda,x}(k)$ and $k^2B_{\Lambda,x}(k)$, respectively.
Note that, by $(\f 2)$, $|B_{\cutoff,x}|\leq \sup_{x}|B_{\infty,x}|\in L^2(\R^d)$
and that $\sup_{x}\|B_{\cutoff,x} - B_{\infty,x}\|\to 0$ as
$\cutoff\rightarrow\infty$. It follows, by a generalization of Lemma~\ref{lm:strongWeyl}, that
\begin{equation*}
   U_{\cutoff}\longrightarrow U_{\infty}\quad (\cutoff\to\infty)
\end{equation*}
strongly in $\hilbert$. To compute $U_{\cutoff} H_{\cutoff} U_{\cutoff}^{*}$
we need the following lemma. From now on, $p$ and $p^2$ often denote $-i\nabla$ and $-\Delta$, respectively.

\begin{lemma}
\label{lm:dressing-parts}
Assume $(\f 1)$ and $(\f 2)$. Then:
\begin{itemize}
\item[(a)] $U_{\cutoff} D(H_0^{1/2}) = D(H_0^{1/2})$ for $\cutoff\leq\infty$ and
$$
U_{\cutoff} p U_{\cutoff}^{*} = p-\phi(kB_{\cutoff})\quad\text{on}\ D(H_0^{1/2}).
$$
\item[(b)] $U_{\cutoff} D(H_0) = D(H_0)$ for $\cutoff<\infty$ and
$$
U_{\cutoff} p^2 U_{\cutoff}^{*}=(p-\phi(kB_{\cutoff}))^2\quad\text{on}\ D(H_0).
$$
\item[(c)] $U_{\infty} D(H_0)$ is a form core of $H_0$.
\end{itemize}
\end{lemma}

Since the components of $p$ are essentially self-adjoint on $D(H_0^{1/2})$, part (a) implies that $U_{\cutoff} p_j U_{\cutoff}^{*} = p_j-\phi(k_j B_{\cutoff})$ as an equality between self-adjoint operators on their respective domains.

\begin{proof} 
(a) Let $\DD= D(H_0)\cap\HH_0$. Then, $\DD$ is an operator core and hence a form core of $H_0$. Moreover, for $\Psi\in\DD$, one shows that 
\begin{equation}\label{eq:pU}
       p U_{\cutoff}^{*}\Psi = U_{\cutoff}^{*}(p-\phi(kB_{\cutoff}))\Psi
\end{equation}
by expanding $U_{\cutoff}^{*}$ in its exponential series. Here we used
that $[\phi(kB_{\cutoff}),\pi(B_{\cutoff})]=0$ by assumption $(\f 1)$ on $v$. Since $\DD$ is a form core of $H_0$ and since $(p-\phi(kB_{\cutoff}))$ is bounded w.r.t. $H_0^{1/2}$ Equation~\eqref{eq:pU} extends to all $\Psi\in D(H_0^{1/2})$ and we see that $U_{\cutoff} D(H_0^{1/2}) \subset D(|p|)$. Since $D(H_0^{1/2}) = D(|p|)\cap D(\sqrt{N})$ and since $D(\sqrt{N})$ is left invariant by $U_{\cutoff}^{*}$, see Lemma~\ref{lm:UNUphi}, we conclude that $U_{\cutoff}^{*}D(H_0^{1/2})\subset D(H_0^{1/2})$. Likewise, $U_{\cutoff} D(H_0^{1/2})\subset D(H_0^{1/2})$ by changing the sign of $v$ and part (a) is proved.

(b) Let $\Psi\in D(H_0)$. Then $U_{\cutoff}^{*}\Psi\in D(H_0^{1/2})$ by
part (a) and $pU_{\cutoff}^{*}\Psi$ is given by Equation~\eqref{eq:pU}. For $\cutoff<\infty$, 
$(p-\phi(kB_{\cutoff}))\Psi\in D(|p|)\cap D(\sqrt{N})
=D(H_0^{1/2})$. It follows, by part (a) again, that $U_{\cutoff}^{*}(p-\phi(kB_{\cutoff}))\Psi\in D(H_0^{1/2})$. Hence, in
view of Equation~\eqref{eq:pU}, $pU_{\cutoff}^{*}\Psi\in D(|p|)$ and
$p^2 U_{\cutoff}^{*}\Psi =
U_{\cutoff}^{*}(p-\phi(kB_{\cutoff}))^2\Psi$. Since
$U_{\cutoff}^{*}D(N)\subset D(N)$ by Lemma~\ref{lm:UNUphi}, part (b)
follows from $D(H_0)=D(p^2)\cap D(N)$.

(c) Let $\HH_1=D(H_0^{1/2})$ equipped with the
form norm of $H_0$. By part (a),  $U_\infty^{*}:\HH_1\to\HH_1$ and
this operator is closed which is easy to see from the continuity of $U_\infty^{*}$ in $\HH$.
Therefore, $U_\infty^{*}$ is bounded in $\HH_1$ by the
closed graph theorem. Since $D(H_0)$ is dense in $\HH_1$, it follows
that $U_\infty^{*}D(H_0)$ is dense in $\HH_1$ as well.
\end{proof}

The results from Lemma~\ref{lm:dressing-parts} (b), the identity
$$
   p\cdot a^*(kB_{\cutoff}) + a(kB_{\cutoff})\cdot p  = a^*(kB_{\cutoff})\cdot p+ p\cdot a(kB_{\cutoff}) -\phi(k^2B_{\cutoff}),
$$
and Lemma~\ref{lm:UNUphi} yield the operator identities 
\begin{eqnarray}
        U_{\cutoff}p^2 U_{\cutoff}^* &=& p^2-2a^*(kB_{\cutoff})\cdot p-2p\cdot a(kB_{\cutoff}) \notag\\
   \label{equ:GTrafoPSquared}   &\phantom{=}&
   +\ \phi(kB_{\cutoff})^2+\phi(k^2B_{\cutoff}) \\
   \label{equ:GTrafoN}        U_{\cutoff}N U_{\cutoff}^* &=& N+\phi(B_{\cutoff})+\norm{B_{\cutoff}}^2 \\
   \label{equ:GTrafoPhi} U_{\cutoff}\phi(G_{\cutoff}) U_{\cutoff}^* &=& \phi(G_{\cutoff})+2\Rea\sprod{B_{\cutoff}}{G_{\cutoff}}
\end{eqnarray}
on $D(H_0)$ for $\cutoff<\infty$. In the above equations, we introduced various
dot-products such as  $p\cdot a(kB_{\cutoff}) = \sum_{j=1}^d p_j a(k_jB_{\cutoff}) $. In view of the fact that
$(1+k^2)B_{\cutoff}=G_{\ir}-G_{\cutoff}$, by definition
\eqref{def:B-cutoff} of $B_{\cutoff}$, we arrive at:

\begin{prop}
\label{proposition:transformedFroehlichHamiltonian}
Assume $(\f 1)$ and $(\f 2)$. Then for all $\cutoff<\infty$, we have 
$U_{\cutoff}H_{\cutoff}U_{\cutoff}^* = H_{\ir} + V_{\ir,\cutoff}$ on $D(H_0)$, where
\begin{align*}
   V_{\ir,\cutoff} &:= -2a^*(kB_{\cutoff})\cdot p-2p\cdot
   a(kB_{\cutoff})+\phi(kB_{\cutoff})^2+C_{\ir,\cutoff},\\
   C_{\ir,\cutoff} &:=\norm{B_{\cutoff}}^2+2\sprod{G_{\cutoff}}{B_{\cutoff}}=\int\limits_{\ir\leq\abs{k}\leq\cutoff}
\abs{v(k)}^2\big( (1+k^2)^{-2} -2(1+k^2)^{-1}\big)\, dk.
\end{align*}
In particular, the operator $H_{\ir,\cutoff}':=H_{\ir}+V_{\ir,\cutoff}$ is
self-adjoint on $D(H_0)$.
\end{prop}

The assumption $(\f 2)$ implies that $kB_{\cutoff}$ is
square integrable even for $\cutoff=\infty$, and hence, the creation and
annihilation operators $a^{*}(kB_{\cutoff})$ and $a(kB_{\cutoff})$
in $V_{\ir,\cutoff}$ are well-defined for $\cutoff=\infty$. Therefore, the first and
the third operators in the sum defining $V_{\ir,\cutoff}$ are well-defined on
$D(H_0)$ for $\cutoff=\infty$. This is not obvious for the second
term, $2p\cdot a(kB_{\cutoff})$, because $p\cdot a(kB_{\cutoff})=
a(kB_{\cutoff})\cdot p + a(k^2B_{\cutoff})$, where the norm of
$k^2B_{\cutoff}$ may diverge as $\Lambda\to\infty$. By imposing,
$(\f 3)$ this problem can be controlled with the help of
Lemma~\ref{lm:FrankSchlein} and we arrive at the following:

\begin{lemma}
\label{lemma:SetInclusionDomainOfH0InDomainOfMomentumOnAnnihilation}
Assume $(\f 1)$ and $(\f 3)$. Then, for all $\ir\leq\cutoff\leq\infty$, the operator $p\cdot a(kB_{\cutoff,x})$ satisfies
$D(p\cdot a(kB_{\cutoff,x}))\supset D(H_0)$ and 
\begin{align*}
   \sup\limits_{\cutoff\leq\infty}\norm{p\cdot
     a(kB_{\cutoff})(H_0+1)^{-1}}\longrightarrow 0\qquad (\ir\to\infty).
\end{align*}
\end{lemma}

\begin{proof}
The operator $p\cdot a(kB_{\cutoff})=\sum\limits_{j=1}^d p_j a(k_jB_{\cutoff})$
is defined on $\bigcap\limits_{j=1}^d D(p_ja(k_jB_{\cutoff}))$. We,
therefore, need to show that $a(k_jB_{\cutoff})\Psi\in D(p_j)$ for
all $\Psi\in D(H_0)$ and all $j=1,2,...,d$. We omit the proof for $\cutoff<\infty$ and only note that for all $\Psi\in D(H_0)$ and $\cutoff<\infty$,
\begin{align}
\label{equ:MomentumOnAnnihilationOperatorForFiniteCutoff}
   p_ja(k_jB_{\cutoff})\Psi=a(k_j^2B_{\cutoff})\Psi+a(k_jB_{\cutoff})p_j\Psi.
\end{align}
The right-hand side is convergent in the limit
$\cutoff\rightarrow\infty$ by
Lemma~\ref{lm:FrankSchlein} and Lemma~\ref{lm:a-like-root-N}. Hence, so is
the left hand side. Since, moreover, 
\begin{align*}
   \limop{\cutoff} a(k_jB_{\cutoff})\Psi=a(k_jB_{\infty})\Psi
\end{align*}
and since $p_j$ is a closed operator, it follows that $a(k_jB_{\infty})\Psi\in D(p_j)$ and that $p_ja(k_jB_{\infty})\Psi$ is given by the limit of~\eqref{equ:MomentumOnAnnihilationOperatorForFiniteCutoff}. This proves that $D(p\cdot a(kB_{\cutoff}))\supset D(H_0)$ and that
\begin{align}\label{equ:MomentumOnAnnihilationOperatorForFiniteCutoffEstimateUp}
      \norm{p_ja(k_jB_{\cutoff})\Psi}\leq \sup\limits_{\cutoff\leq\infty}\left(\|a(k_j^2B_{\cutoff})\Psi\|+\|a(k_jB_{\cutoff})p_j\Psi\|\right).
\end{align}
Using Lemma~\ref{lm:a-like-root-N} and Lemma~\ref{lm:FrankSchlein}, it is easy to see that
\begin{align}
\label{equ:MomentumOnAnnihilationOperatorEstimateUpFirstSummand}
   \norm{a(k_jB_{\infty})p_j\Psi}\leq\left(\int\limits_{\abs{k}\geq\ir}k^2\frac{\abs{v(k)}^2}{(1+k^2)^2}dk\right)^{1/2}\norm{H_0\Psi}
\end{align}
and
\begin{align}
\label{equ:MomentumOnAnnihilationOperatorEstimateUpSecondSummand}
   \norm{a(k_j^2B_{\infty})\Psi}\leq\left(\sup\limits_{q\in\R^d}\int\limits_{\abs{k}\geq\ir}\frac{\abs{v(k)}^2}{1+(q-k)^2}dk\right)^{1/2}\norm{(H_0+1)\Psi}.
\end{align}
Upon combining Inequalities~\eqref{equ:MomentumOnAnnihilationOperatorForFiniteCutoffEstimateUp}, \eqref{equ:MomentumOnAnnihilationOperatorEstimateUpFirstSummand} and \eqref{equ:MomentumOnAnnihilationOperatorEstimateUpSecondSummand}, the second assertion of the lemma follows. 
\end{proof}

\begin{theorem}
\label{theorem:selfadjointnessOfTheTransformedHamiltonian}
Assume $(\f 1)$ and $(\f 3)$. Then, for every $\eps>0$, there exist $\ir>0$ and $C_{\eps}\in\R$ such that for all $\cutoff\leq\infty$ and all $\Psi\in D(H_0)$,
\begin{align}
\label{eq:Kato-small}
   \norm{V_{\ir,\cutoff}\Psi}\leq\eps\norm{H_0\Psi}+C_{\eps}\norm{\Psi}.
\end{align}
The operator $H_{\ir,\infty}'=H_{\ir}+V_{\ir,\infty}$ is self-adjoint on $D(H_0)$ provided $\ir$ is large enough.
\end{theorem}

\begin{proof}
It suffices to establish the desired estimate for each term in the sum
\begin{align}\label{eq:Kato-small-V}
   V_{\ir,\cutoff}=\phi(kB_{\cutoff})^2-2a^*(kB_{\cutoff})\cdot p-2p\cdot a(kB_{\cutoff})+C_{\ir,\cutoff}.
\end{align}
By Lemmas~\ref{lm:a-like-root-N} and~\ref{lemma:standardEstimateForSquaredSegalOperator}, for all $\Psi\in D(H_0)$,
\begin{align*}
      \norm{a^*(kB_{\cutoff})\cdot p\Psi} &\leq\norm{kB_{\cutoff}}\ \norm{\sqrt{N+1}p\Psi}, \\
      \norm{\phi(kB_{\cutoff})^2\Psi} &\leq 4\sqrt{2}\norm{kB_{\cutoff}}^2\ \norm{(N+1)\Psi},
\end{align*}
where $N+1$ and $\sqrt{N+1}p$ are $H_0$-bounded and 
\begin{align*}
   \norm{kB_{\cutoff}}^2\leq\int\limits_{\abs{k}\geq\ir}k^2\frac{\abs{v(k)}^2}{(1+k^2)^2}dk\to 0\qquad (\ir\to\infty).
\end{align*}
This proves Inequality~\eqref{eq:Kato-small} as far as the first two terms in~\eqref{eq:Kato-small-V} are concerned. For the operator $p\cdot a(kB_{\cutoff})$, the desired estimate follows from Lemma~\ref{lemma:SetInclusionDomainOfH0InDomainOfMomentumOnAnnihilation} and $C_{\ir,\cutoff}$ is bounded uniformly in $\cutoff$.

In view of~\eqref{eq:Kato-small}, the self-adjointness follows from Kato-Rellich because $H_{\ir}=H_0+\phi(G_{\ir})$ where $\phi(G_{\ir})$ is infinitesimally $H_0$-bounded for every given fixed $\ir$.
\end{proof}

\begin{prop}
\label{prop:dressed-R-converge}
Assume $(\f 1)$ and $(\f 3)$. Then, for $\ir$ sufficiently large, $H_{\ir,\cutoff}'\rightarrow H_{\ir,\infty}'$ in the norm resolvent sense as $\cutoff\rightarrow\infty$.
\end{prop}

\begin{proof}
For short, we set $H_{\cutoff}':=H_{\ir,\cutoff}'$ in this proof. By Theorem~\ref{theorem:selfadjointnessOfTheTransformedHamiltonian}, $H_{\cutoff}'$ is self-adjoint on $D(H_0)$ for all $\cutoff\leq\infty$ if $\ir$ is sufficiently large. In view of Theorem VIII.25(b) from~\cite{ReedSimonI1972}, it therefore suffices to prove that
\begin{equation*}
   (H_{\cutoff}'-H_{\infty}')(H_{\infty}'+i)^{-1}\rightarrow 0\qquad(\cutoff\to\infty)
\end{equation*}
which is equivalent to
\begin{equation}\label{eq:nrc-prime}
   (H_{\cutoff}'-H_{\infty}')(H_{0}+i)^{-1}\rightarrow 0\qquad(\cutoff\to\infty)
\end{equation}
due to the boundedness of $(H_0+i)(H_{\infty}'+i)^{-1}$. By definition of $H_{\cutoff}'$, see Proposition~\ref{proposition:transformedFroehlichHamiltonian}, we have
\begin{align}
   H_{\cutoff}'-H_{\infty}'=&\ V_{\ir,\cutoff}-V_{\ir,\infty} \notag \\
                       =&-\phi(kB_{\cutoff})\phi(kB_{\infty}\chibar_{\cutoff})-\phi(kB_{\infty}\chibar_{\cutoff})\phi(kB_{\infty}) \notag \\
   \label{equ:differenceBetweenTheTransformedHamiltoniansWithAndWithoutCutoff}   &+2a^*(kB_{\infty}\chibar_{\cutoff})\cdot p+2p\cdot a(kB_{\infty}\chibar_{\cutoff})+C_{\cutoff,\infty}.
\end{align}
Here, $\chibar_{\cutoff}:=1-\chi_{\cutoff}$, and we used $B_{\cutoff,x}-B_{\infty,x}=-B_{\infty,x}\chibar_{\cutoff}$. Convergence~\eqref{eq:nrc-prime} now follows from~\eqref{equ:differenceBetweenTheTransformedHamiltoniansWithAndWithoutCutoff} by the same estimates that were used in the proof of Theorem~\ref{theorem:selfadjointnessOfTheTransformedHamiltonian}.
\end{proof}


\begin{corollary}
\label{corollary:relativeBoundednessOfTheTransformedHamiltonianWithRespectToTheFreeHamiltonian}
Assume $(\f 1)$ and $(\f 3)$. Then for $\ir$ sufficiently large, there exists a constant $C$ such that for all $\cutoff\leq\infty$,
$$
   \frac{1}{2}H_0-C\leq H_{K,\cutoff}'\leq \frac{3}{2}H_0+C.
$$
\end{corollary}

\begin{proof}
By Theorem~\ref{theorem:selfadjointnessOfTheTransformedHamiltonian}, there exist $\ir$ and $C$ such that
$$
   \norm{(V_{\ir,\cutoff}+\phi(G_{\ir}))\Psi}\leq \frac{1}{4}\norm{H_0\Psi}+\frac{C}{2}\norm{\Psi}
$$
for all $\Psi\in D(H_0)$ and $\cutoff\leq\infty$. Using the lower bound given by the Kato-Rellich-theorem, Theorem X.12 from~\cite{ReedSimonII1975}, we conclude that
$$
   \frac{1}{2}H_0\pm \left(\phi(G_{\ir})+V_{\ir,\cutoff}\right)\geq -C, 
$$
which implies the desired inequalities.
\end{proof}

\begin{theorem}
\label{thm:domain-H}
Assume $(\f 1)$ and $(\f 3)$. Then, there exists a self-adjoint operator $H$ such that $H_{\cutoff}\rightarrow H$ as $\cutoff\rightarrow\infty$ in the norm resolvent sense. This operator has the representation
\begin{align*}
   H=U_{\ir,\infty}^*H_{\ir,\infty}'U_{\ir,\infty}, \quad D(H)=U_{\ir,\infty}^*D(H_0),
\end{align*}
which is valid for $\ir$ sufficiently large. If $D\subset D(H_0)$ is a core of $H_0$, then $U_{\ir,\infty}^*D$ is a core of $H$.
\end{theorem}

In this theorem, $U_{\ir,\infty} = U_{\infty}$ to exhibit the dependence of $U_{\infty}$ on $K$.
The theorem implies, in particular, that $H:=U_{\ir,\infty}^*H_{\ir,\infty}'U_{\ir,\infty}$
is independent of $\ir$ for $\ir$ sufficiently large. Because of the convergence
$H_{\cutoff}\rightarrow H$, this operator coincides with the operator constructed
in Section~\ref{sec:FroehlichHamiltonian}.

\begin{proof}
Choose $\ir$ so large that $H_{\ir,\cutoff}'\rightarrow H_{\ir,\infty}'$ in the norm resolvent 
sense by Proposition~\ref{prop:dressed-R-converge}. In the following, $\ir$ is fixed and 
suppressed. Let $R_{\cutoff}'(z):=(H_{\cutoff}'-z)^{-1}$ and 
$H:=U_{\infty}^*H_{\infty}'U_{\infty}$ in this proof. By Proposition~\ref{proposition:transformedFroehlichHamiltonian}, 
$H_{\cutoff}=U_{\cutoff}^*H_{\cutoff}'U_{\cutoff}$ for all $\cutoff<\infty$ and, therefore,
\begin{align}
\label{equ:NRC-FH-GT-0}
   (H_{\cutoff}-z)^{-1}-(H-z)^{-1}=& \ U_{\cutoff}^*R_{\cutoff}'(z)U_{\cutoff}-U_{\infty}^*R_{\infty}'(z)U_{\infty} \notag \\
                   =& \ (U_{\cutoff}^*-U_{\infty}^*)R_{\cutoff}'(z)U_{\cutoff}+U_{\infty}^*(R_{\cutoff}'(z)-R_{\infty}'(z))U_{\cutoff} \notag \\
                               &+U_{\infty}^*R_{\infty}'(z)(U_{\cutoff}-U_{\infty}).
\end{align}
It remains to show that these three terms vanish in the limit $\cutoff\rightarrow\infty$. For the second term, this follows from Proposition~\ref{prop:dressed-R-converge}. For the first and third terms, we have
\begin{align}
\label{equ:NRC-FH-GT-1}
   \norm{(U_{\cutoff}^*-U_{\infty}^*)R_{\cutoff}'(z)U_{\cutoff}}\leq\norm{(U_{\cutoff}^*-U_{\infty}^*)(H_0+1)^{-1/2}}\cdot\norm{(H_0+1)^{1/2}R_{\cutoff}'(z)}
\end{align}
and
\begin{align}
   \norm{U_{\infty}^*R_{\infty}'(z)(U_{\cutoff}-U_{\infty})}&=\norm{(U_{\cutoff}^*-U_{\infty}^*)R_{\infty}'(\overline{z})}\nonumber\\
               &\leq\norm{(U_{\cutoff}^*-U_{\infty}^*)(H_0+1)^{-1/2}}\cdot\norm{(H_0+1)^{1/2}R_{\infty}'(\overline{z})}.\label{equ:NRC-FH-GT-2}
\end{align}
Lemma~\ref{lm:strongWeyl} implies that
\begin{align}
\label{equ:NRC-FH-GT-3}
   \norm{(U_{\cutoff}^*-U_{\infty}^*)(H_0+1)^{-1/2}}\rightarrow 0 \qquad (\cutoff\rightarrow\infty)
\end{align}
because $\sup\limits_{x\in\R^d}\norm{B_{\cutoff,x}-B_{\infty,x}}\rightarrow 0$ as $\cutoff\rightarrow\infty$, and the Corollary~\ref{corollary:relativeBoundednessOfTheTransformedHamiltonianWithRespectToTheFreeHamiltonian} shows that
\begin{align}
\label{equ:equ:NRC-FH-GT-4}
   \sup\limits_{\cutoff\leq\infty}\norm{(H_0+1)^{1/2}R_{\cutoff}'(z)}<\infty.
\end{align}
Combining Properties~\eqref{equ:NRC-FH-GT-1}, \eqref{equ:NRC-FH-GT-2}, \eqref{equ:NRC-FH-GT-3}, and~\eqref{equ:equ:NRC-FH-GT-4} we see that the first and third terms from~\eqref{equ:NRC-FH-GT-0} vanish as $\cutoff\rightarrow\infty$. The statement about $D(H)$ follows from Theorem~\ref{theorem:selfadjointnessOfTheTransformedHamiltonian}.

Now, if $D\subset D(H_0)$ is a core of $H_0$, then, by Theorem~\ref{theorem:selfadjointnessOfTheTransformedHamiltonian}, $D$ is a core of $H_{\infty}'$, and hence, $U_{\infty}^*D$ is a core of $H=U_{\infty}^*H_{\infty}'U_{\infty}$.
\end{proof}

\section{Regularity of domain vectors}

\label{sec:regularity}
In this section, we prove Equations \eqref{D-supset} and \eqref{D-cap}
of the Introduction. As a preparation we need the following lemma, which generalizes the statement of Lemma~\ref{lm:dressing-parts} (a), on the invariance of $D(H_0^{1/2}) = D(\abs{p})\cap D(N^{1/2})$ under the transformation $U_{\infty}$.

\begin{lemma}
\label{lm:invariance-of-p-domains}
Assume $(\f 1)$ and $(\f 2)$. Then, for $\sigma\in [0,1]$, the subspaces $D(\abs{p}^{\sigma})\cap D(N^{\sigma/2})$ are left invariant by $U_{\infty}$ and $U_{\infty}^{*}$. 
\end{lemma}

\begin{proof}
Let $\gamma\in D(\abs{p}^{\sigma})\cap D(N^{\sigma/2})$. Then, $U_{\infty}\gamma\in D(N^{\sigma/2})$, by Lemma~\ref{lm:UNUphi}, and it remains to prove that  $U_{\infty}\gamma\in D(\abs{p}^{\sigma})$. Since $U_{\infty}^*pU_{\infty}=p+\phi(kB_{\infty})$, by Lemma~\ref{lm:dressing-parts} and the remark thereafter, $U_{\infty}\gamma\in D(\abs{p}^{\sigma})$ is equivalent to $\gamma\in D(\abs{p+\phi(kB_{\infty})}^{\sigma})$. To prove the latter, we first observe that $\phi(kB_{\infty})^2 \leq C(2N+1)$, where $C=2\norm{kB_{\infty}}^2$, and hence,
$$
   (p+\phi(kB_{\infty}))^2 \leq 2p^2+2\phi(kB_{\infty})^2 \leq 2(p^2 +C (2N+1)).
$$
This means, in particular, that the form domain of $(p+\phi(kB_{\infty}))^2$ contains the form domain of $p^2+N$. 
From the operator monotonicity of the fractional power $\sigma$ (see \cite{Schmuedgen}, Proposition~10.14), it follows that
\begin{align}
   \abs{p+\phi(kB_{\infty})}^{2\sigma} &\leq 2^{\sigma}\left(\abs{p}^{2\sigma}+ C^{\sigma}(2N+1)^{\sigma}\right).\label{eq:monotone}
\end{align}
Inequality \eqref{eq:monotone} again includes a statement about form domains. It implies that 
$D(\abs{p+\phi(kB_{\infty})}^{\sigma})\supset D(|p|^{\sigma})\cap D(N^{\sigma/2})$. In view of the assumption on $\gamma$, this is exactly what we needed to show. 
\end{proof}


\begin{theorem}\label{thm:domain-intersect}
\label{theorem:generalizationOfTheNonInvarianceOfTheDomainWithoutCutoff}
Assume $(\f 1)$ and $(\f 2)$. If $\int\abs{v(k)}^2(1+k^2)^{s-2}dk=\infty$ for some $s\in (1,2]$, then 
$$
   U_{\infty}^*D(H_0)\cap D((-\Delta)^{s/2})=\{0\}.
$$
\end{theorem}

In the case of the Fr\"ohlich Hamiltonian where $v(k) =|k|^{-(d-1)/2}$, choose $s=3/2$ in Theorem~\ref{thm:domain-intersect}
to prove the Assertion \eqref{D-cap} in the Introduction.

\begin{proof}
Let $\Psi\in D(H_0)$, and suppose that $U_{\infty}^*\Psi\in D((-\Delta)^{s/2})=D(\abs{p}^s)$. Then, $U_{\infty}^*\Psi\in D(\abs{p}^s)\cap D(N)$ by Lemma~\ref{lm:UNUphi}.
In view of the inequality
\begin{align*}
   N^{(s-1)/2}\abs{p}\leq \frac{s-1}{s}N^{s/2}+\frac{1}{s}\abs{p}^s, 
\end{align*}
and the assumption $1<s\leq 2$, we conclude that $pU_{\infty}^*\Psi\in D(\abs{p}^{s-1})\cap D(N^{(s-1)/2})$. This implies, by Lemma~\ref{lm:invariance-of-p-domains}, that $U_{\infty}pU_{\infty}^*\Psi\in D(\abs{p}^{s-1})$, where
\begin{align*}
   U_{\infty}pU_{\infty}^*\Psi=\left(p-\phi(kB_{\infty})\right)\Psi.
\end{align*}
The first term on the right-hand side, $p\Psi$, belongs to $D(\abs{p}^{s-1})$ as well, because $s\leq 2$ and $\Psi\in D(p^2)$. We now compute $\norm{\abs{p}^{s-1}\phi(kB_{\infty})\Psi}$, and we show that this number is infinite unless $\Psi=0$. To this end, we define the functions
\begin{align*}
   \abs{p}_{\eps}^{s-1}&:=\frac{\abs{p}^{s-1}}{1+\eps\abs{p}^{s-1}}, \\
   D_{\eps}(p,k)&:=\abs{p+k}_{\eps}^{s-1}-\abs{k}_{\eps}^{s-1}
\end{align*}
for $p,k\in\R^d$ and $\eps>0$. Using that $(s-1)\in (0,1]$, it is straightforward to verify that
\begin{align}
\label{equ:inequalityForTheRegularizationFunctionForMomentumWithSomeExponent2}
   \abs{D_{\eps}(p,k)}\leq \abs{p}^{s-1}   
\end{align}
for all $p,k\in\R^d$ and $\eps>0$. For $p=-i\nabla_x$, we have 
$$
   \abs{p}_{\eps}^{s-1}e^{ikx}=e^{ikx}\abs{p+k}_{\eps}^{s-1}
$$
which, in view of \eqref{def:annihilation}, implies
\begin{align*}
   \abs{p}_{\eps}^{s-1}a(kB_{\infty})\Psi&=a(kB_{\infty}\abs{p}_{\eps}^{s-1})\Psi \notag \\
                                  &=a(\abs{k}_{\eps}^{s-1}kB_{\infty})\Psi+a(D_{\eps}(p,k)kB_{\infty})\Psi
\end{align*}
and
$$
   \abs{p}_{\eps}^{s-1}a^*(kB_{\infty})\Psi=a^*(kB_{\infty}\abs{k}_{\eps}^{s-1})\Psi+a^*(kB_{\infty}D_{\eps}(p,-k))\Psi.
$$
By Inequality~\eqref{equ:inequalityForTheRegularizationFunctionForMomentumWithSomeExponent2},
\begin{align}
   \label{equ:estimateOfTheTermWithTheAnnihilatorOfTheBadTerm}   \norm{a(D_{\eps}(p,k)kB_{\infty})\Psi} &\leq \norm{kB_{\infty}}\cdot\norm{\abs{p}^{s-1}\sqrt{N}\Psi}, \\
   \label{equ:estimateOfTheTermWithTheCreatorOfTheBadTerm}       \norm{a^*(kB_{\infty}D_{\eps}(p,-k))\Psi} &\leq \norm{kB_{\infty}}\cdot\norm{\abs{p}^{s-1}\sqrt{N+1}\Psi}, 
\end{align}
and by Lemma~\ref{lm:FrankSchlein},
\begin{align}
\label{equ:anotherAplicationOfTheFrankSchleinLemma}
   \norm{a(kB_{\infty}\abs{k}_{\eps}^{s-1})\Psi} \leq C_s\norm{\sqrt{N}(1+p^2)^{1/2}\Psi}, 
\end{align}
where
\begin{equation*}
   C_s:=\sup\limits_{q\in\R^d}\left(\int\frac{\abs{kB_{\infty}(k)}^2\abs{k}^{2(s-1)}}{1+(q-k)^2}dk\right)<\infty
\end{equation*}
because $s\leq 2$. Note that the bounds~\eqref{equ:estimateOfTheTermWithTheAnnihilatorOfTheBadTerm}, \eqref{equ:estimateOfTheTermWithTheCreatorOfTheBadTerm} and~\eqref{equ:anotherAplicationOfTheFrankSchleinLemma} are uniform in $\eps>0$. Therefore, there exists a constant $c$ such that
\begin{align*}
   \norm{\abs{p}^{s-1}\phi(kB_{\infty})\Psi}&=\limop[0]{\eps}\norm{\abs{p}_{\eps}^{s-1}\phi(kB_{\infty})\Psi} \\
                                          &\geq\liminfop[0]{\eps}\norm{a^*(kB_{\infty}\abs{k}_{\eps}^{s-1})\Psi}-c \\
                                          &\geq\limop[0]{\eps}\norm{kB_{\infty}\abs{k}_{\eps}^{s-1}}\cdot\norm{\Psi}-c
\end{align*}
which is infinite unless $\Psi=0$. This completes the proof.
\end{proof}

\begin{theorem}
\label{thm:domain-supset}
Assume $(\f 1)$. If $\int\abs{v(k)}^2(1+k^2)^{s-2}dk<\infty$ for some $s\in [1,2]$, then 
$$
   U_{\infty}^*D(H_0)\subset D((-\Delta)^{s/2}).
$$
\end{theorem}

In the case of the Fr\"ohlich Hamiltonian where $v(k) = |k|^{-(d-1)/2}$, the
assumption of Theorem~\ref{thm:domain-supset}  is satisfied for all
$s\in [1,3/2)$ and $U_{\infty}^*D(H_0)=D(H)$ by Theorem~\ref{thm:domain-H}. This proves Assertion~\eqref{D-supset} in the
Introduction.

\begin{proof}
Let $\Psi\in D(H_0)$. From Lemma~\ref{lm:dressing-parts} we know that $U_{\infty}^*\Psi\in D(H_0^{1/2})\subset D(\abs{p})$ and that $p_jU_{\infty}^*\Psi=U_{\infty}^*(p_j-\phi(k_jB_{\infty}))\Psi=:U_{\infty}^*\gamma_j$. It follows that
$$
   \norm{\abs{p}^sU_{\infty}^*\Psi}^2=\sum\limits_{j=1}^d\norm{\abs{p}^{s-1}p_jU_{\infty}^*\Psi}^2=\sum\limits_{j=1}^d\norm{\abs{p}^{s-1}U_{\infty}^*\gamma_j}^2,
$$
which is finite (and thus proves the theorem) provided we can show that $U_{\infty}^*\gamma_j\in D(|p|^{s-1})$ for all $j$.
To prove this, it suffices, by Lemma~\ref{lm:invariance-of-p-domains},  to show that $\gamma_j\in D(|p|^{s-1})\cap D(N^{(s-1)/2})$. From $\gamma_j\in D(\sqrt{N})$ it follows that $\gamma_j\in D(N^{(s-1)/2})$ because $s\in [1,2]$. It remains to show that $\||p|^{s-1}\gamma_j\|<\infty$.
The first term of $\gamma_j=p_j\Psi-\phi(k_jB_{\infty})\Psi$ belongs to $D(\abs{p}^{s-1})$ because $s\leq 2$ and because $\Psi\in D(-\Delta)$. To prove that $\|\abs{p}^{s-1}\phi(k_jB_{\infty})\Psi\|$ is finite, we recall the estimates in the proof of Theorem~\ref{theorem:generalizationOfTheNonInvarianceOfTheDomainWithoutCutoff} which imply that $\|\abs{p}^{s-1} a(k_jB_{\infty})\Psi\|$ is finite and that
\begin{align}
\label{eq:domain-subset} 
   \norm{\abs{p}^{s-1} a^*(k_jB_{\infty})\Psi}&\leq \norm{a^*(k_jB_{\infty}\abs{k}^{s-1})\Psi}+\norm{a^*(k_jB_{\infty}D_{\eps}(p,-k))\Psi} \notag \\
                                           &\leq C_s\left(\norm{\sqrt{N+1}\Psi}+\norm{\sqrt{N+1}\abs{p}^{s-1}\Psi}\right),
\end{align}
where
$$
   C_s^2:=\int\abs{k}^{2s}\frac{\abs{v(k)}^2}{(1+k^2)^2}dk<\infty.
$$
In the last inequality, we used the hypothesis on $v$, and in~\eqref{eq:domain-subset}, we used $s\geq 1$.
\end{proof}


\appendix
\section{Quadratic forms and resolvent convergence}   
\label{app:abstract}

The following theorem is our main tool for the proof of
Theorem~\ref{thm:nrs}. It is essentially due to Nelson~\cite{Nelson1964}. A
similar theorem, without proof, is given in the Appendix of \cite{AmmariZied2000}. 

\begin{theorem}
\label{thm:Ammari}
Let $H_0\geq 0$ be a self-adjoint operator in $\hilbert$ and let $\norm{\Psi}_0:=\norm{(H_0+1)^{1/2}\Psi}$ for $\Psi\in D(H_0^{1/2})$. For each $\cutoff<\infty$, let $W_{\cutoff}$ be a quadratic form defined on $D(H_0^{1/2})$ such that
\begin{itemize}
\item[(a)] for all $\Psi\in D(H_0^{1/2})$ and all $\cutoff<\infty$,
$$
   \abs{W_{\cutoff}(\Psi)}\leq a\norm{\Psi}_0^2+b_{\cutoff}\norm{\Psi}^2,
$$
where $a<1$,
\item[(b)] for all $\Psi\in D(H_0^{1/2})$,
$$
      \abs{W_{\cutoff}(\Psi)-W_{\cutoff'}(\Psi)}\leq C_{\cutoff,\cutoff'}\norm{\Psi}_0^2,
$$
where $C_{\cutoff,\cutoff'}\rightarrow 0$ as $\cutoff,\cutoff'\rightarrow\infty$.
\end{itemize}
Let $W_{\infty}(\Phi,\Psi):=\limop{\cutoff}W_{\cutoff}(\Phi,\Psi)$. Then, (a) extends to $\cutoff=\infty$ with some finite $b_{\infty}$, and for each $\cutoff\leq\infty$, there exists a self-adjoint, semibounded operator $H_{\cutoff}$ with $D(H_{\cutoff})\subset D(H_0^{1/2})$ and 
\begin{align}
\label{equ:AmariLikeTheoremStatement1} 
      \sprod{\Phi}{H_{\cutoff}\Psi}=\sprod{H_0^{1/2}\Phi}{H_0^{1/2}\Psi}+W_{\cutoff}(\Phi,\Psi)
\end{align}
for all $\Phi\in D(H_0^{1/2})$ and $\Psi\in D(H_{\cutoff})$. Moreover, for all $z\in\C\backslash\R$,
\begin{align*}
      (H_{\cutoff}-z)^{-1} \longrightarrow (H_{\infty}-z)^{-1}\qquad (\Lambda\to\infty)
\end{align*}
in the operator norm.
\end{theorem}

\begin{proof}
Choose $\cutoff_0>0$ so large, that $C_{\cutoff,\cutoff'}\leq (1-a)/2$ for all $\cutoff,\cutoff'\geq\cutoff_0$. Then, for $\cutoff\geq\cutoff_0$,
\begin{align}
   \abs{W_{\cutoff}(\Psi)} &\leq \abs{W_{\cutoff}(\Psi)-W_{\cutoff_0}(\Psi)}+\abs{W_{\cutoff_0}(\Psi)} \notag \\
                         &\leq C_{\cutoff,\cutoff_0}\norm{\Psi}_0^2+a\norm{\Psi}_0^2+b_{\cutoff_0}\norm{\Psi}^2 \notag \\
   \label{equ:AmariLikeTheoremProofFirstEstimate} &\leq \frac{1}{2}(1+a)\norm{\Psi}_0^2+b_{\cutoff_0}\norm{\Psi}^2.
\end{align}
In the limit $\cutoff\rightarrow\infty$, it follows that
\begin{align}
\label{equ:AmariLikeTheoremProofFollowingOfTheFirstEstimate} 
   \abs{W_{\infty}(\Psi)} \leq \frac{1}{2}(1+a)\norm{\Psi}_0^2+b_{\cutoff_0}\norm{\Psi}^2.
\end{align}
From Assumption~(a) and from~\eqref{equ:AmariLikeTheoremProofFollowingOfTheFirstEstimate}, it follows that for each $\cutoff\leq\infty$, the quadratic form
\begin{align*}
   \sprod{H_0^{1/2}\Phi}{H_0^{1/2}\Psi}+W_{\cutoff}(\Phi,\Psi)
\end{align*}
with $\Phi,\Psi\in D(H_0^{1/2})$ is closed, bounded from below and hence associated with a unique self-adjoint operator $H_{\cutoff}$ such that~\eqref{equ:AmariLikeTheoremStatement1} holds (see~\cite{KatoPerturbation1976}). The Inequalities~\eqref{equ:AmariLikeTheoremProofFirstEstimate} and~\eqref{equ:AmariLikeTheoremProofFollowingOfTheFirstEstimate} imply that
\begin{align}
\label{equ:AmariLikeTheoremProofImplicationOfTheTwoFirstInequalities} 
   H_0\leq\frac{2}{1-a}(H_{\cutoff}+M), \quad \cutoff_0\leq\cutoff\leq\infty,
\end{align}
where $M:=b_{\cutoff_0}+1$. By Assumption~(b), $C_{\cutoff}:=\limsupop{\cutoff'} C_{\cutoff,\cutoff'}\rightarrow 0$ as $\cutoff\rightarrow\infty$ and
\begin{align}
\label{equ:AmariLikeTheoremProofImplicationOfTheSecondAssumption} 
   \abs{W_{\cutoff}(\Psi)-W_{\infty}(\Psi)} = \limop{\cutoff'}\abs{W_{\cutoff}(\Psi)-W_{\cutoff'}(\Psi)}\leq C_{\cutoff}\norm{\Psi}_0^2.
\end{align}
Using~\eqref{equ:AmariLikeTheoremProofImplicationOfTheSecondAssumption}, we conclude that for all $\cutoff\geq\cutoff_0$ and all $\Phi,\Psi\in\hilbert$,
\begin{align*}
   &\abs{\sprod{\Phi}{\left(R_{\cutoff}(z)-R_{\infty}(z)\right)\Psi}} \notag \\
                                                      &= \abs{\sprod{R_{\cutoff}(\overline{z})\Phi}{(H_{\infty}-z)R_{\infty}(z)\Psi}-\sprod{(H_{\cutoff}-\overline{z})R_{\cutoff}(\overline{z})\Phi}{R_{\infty}(z)\Psi}} \notag \\
                                                      &= \abs{W_{\infty}(R_{\cutoff}(\overline{z})\Phi,R_{\infty}(z)\Psi)-W_{\cutoff}(R_{\cutoff}(\overline{z})\Phi,R_{\infty}(z)\Psi)} \notag \\
                                                     &\leq C_{\cutoff}\norm{R_{\cutoff}(\overline{z})\Phi}_0\norm{R_{\infty}(z)\Psi}_0,
\end{align*}
where by~\eqref{equ:AmariLikeTheoremProofImplicationOfTheTwoFirstInequalities}, $\norm{R_{\cutoff}(\overline{z})\Phi}_0\leq C_z\norm{\Phi}$ and $\norm{R_{\infty}(z)\Psi}_0\leq C_z\norm{\Psi}$ with $C_z$ independent of $\cutoff$ for $\cutoff_0\leq\cutoff\leq\infty$.
\end{proof}

\section{Creation and annihilation operators}   
\label{app:creation}

Let $\fock=\bigoplus\limits_{n\geq 0}\fock_n$ be the symmetric Fock space over some Hilbert space $\hh$, let $\hilbert=L^2(\R^d)\otimes\fock$, and let $\hilbert_0=\bigcup\limits_{n\geq 0}\chi(N\leq n)\hilbert$. Suppose
\begin{align*}
   B: L^2(\R^d)\rightarrow L^2(\R^d)\otimes\hh
\end{align*}
is a bounded linear operator. Then, we define the operator $a^*(B)$ in $\hilbert$ on vectors $\Psi=(\Psi^{(n)})_{n\geq 0}\in\hilbert$ by
\begin{align}
\label{def:creation}
   a^*(B)\Psi^{(n)}=\sqrt{n+1}\ S_{n+1}(B\otimes\ID)\Psi^{(n)},
\end{align}
where $S_{n+1}$ denotes the orthogonal projection from $\bigotimes^{n+1}\hh$ onto $\fock_{n+1}$. The annihilation operator $a(B)$ is defined on $\hilbert_0$ by $\sprod{a(B)\Phi}{\Psi}=\sprod{\Phi}{a^*(B)\Psi}$ for all $\Phi,\Psi\in\hilbert_0$. One easily verifies that $a(B)=0$ on $L^2(\R^d)\otimes\fock_{n=0}$ and that
\begin{align}
\label{def:annihilation}
   a(B)\Psi^{(n)}=\sqrt{n}\ (B^*\otimes\ID)\Psi^{(n)}.
\end{align}
Since $\hilbert_0$ is dense in $\hilbert$, it follows that both $a(B)$ and $a^*(B)$ are closable, and we denote the closures by $a(B)$ and $a^*(B)$ as well. It is straightforward to show that $a^*(B)$ is the adjoint of $a(B)$, see for example~\cite{BratteliRobinsonII}.

The following lemma easily follows from~\eqref{def:creation} and~\eqref{def:annihilation}:
\begin{lemma}
\label{lm:a-like-root-N}
Let $B:L^2(\R^d)\to L^2(\R^d)\otimes \hh$ be a bounded linear operator with norm $\norm{B}$. Then, $D(a^{\#}(B))\supset D(\sqrt{N})$, and for all $\Psi\in D(\sqrt{N})$,
\begin{align*}
      \norm{a(B)\Psi}   &\leq \norm{B}\ \norm{\sqrt{N}\ \Psi}, \\
      \norm{a^*(B)\Psi} &\leq \norm{B}\ \norm{\sqrt{N+1}\ \Psi}.
\end{align*}
\end{lemma}

Creation and annihilation operators $a^*(f)$ and $a(f)$ for $f\in\hh$ are defined in terms of the linear operator from $L^2(\R^d)$ to $L^2(\R^d)\otimes\hh$ which maps $\Psi$ to $\Psi\otimes f$. The norm of this operator is $\norm{f}$. Lemma~\ref{lm:a-like-root-N}, therefore, implies that for all $f\in\hh$ and all $\Psi\in D(\sqrt{N})$,
\begin{align*}
     \norm{a(f)\Psi}   &\leq \norm{f}\ \norm{\sqrt{N}\ \Psi}, \\
     \norm{a^*(f)\Psi} &\leq \norm{f}\ \norm{\sqrt{N+1}\ \Psi}.
\end{align*}
From these Estimates and from the pull-through formulas $a(f)N=(N+1)a(f)$ and $Na^*(f)=a^*(f)(N+1)$ the
next two lemmas follow easily.
\begin{lemma}
\label{lemma:standardEstimateForSquaredSegalOperator}
For all $f,g\in L^2(\R^d)$,
\begin{align*} 
   \norm{a^{\#}(f)a^{\#}(g)(N+1)^{-1}} &\leq \sqrt{2}\norm{f}\
   \norm{g},\\
    \norm{\phi(f)^2(N+1)^{-1}}  &\leq 4\sqrt{2}\norm{f}^2,
\end{align*}
where $a^{\#}$ stands for $a$ or $a^*$.
\end{lemma}

\begin{lemma}\label{lm:UNUphi}
Let $f,g\in L^2(\R^d)$. Then, the domains of $\phi(g)$ and $N^{\sigma}$, for $\sigma\in [0,1]$,
are left invariant by $e^{-i\pi(f)} $ and
\begin{align*}
  e^{i\pi(f)} \phi(g) e^{-i\pi(f)} &=  \phi(g) + 2\Rea\sprod{f}{g}\quad \text{on}\ D(\phi(g)), \\   
 e^{i\pi(f)} N e^{-i\pi(f)} &= N + \phi(f) +\|f\|^2\quad \text{on}\
 D(N).
\end{align*}
\end{lemma}

\begin{proof}
For the first equation including the statement on the domain of
$\phi(g)$, see Proposition~5.2.4 of \cite{BratteliRobinsonII}. The
method of proof of this proposition in \cite{BratteliRobinsonII} can be
generalized to prove the invariance of $D(N)$ and the second equation. The invariance of $D(N^\sigma)$ for $\sigma\in (0,1)$ now follows by a simple interpolation argument based on the Hadamard three-lines theorem. 
\end{proof}

\begin{lemma}
\label{lm:strongWeyl}
Let $f,g\in L^2(\R^d)$. Then
\begin{align}
\label{equ:standardEstimateForDifferenceOfWeylOperators} 
   \norm{(e^{i\pi(f)}-e^{i\pi(g)})(N+1)^{-1/2}}\leq 2\norm{f-g}+\abs{\Ima\sprod{f}{g}}.  
\end{align}
\end{lemma}

\begin{proof}
For any $\Psi\in D(N^{1/2})$, we have
\begin{align*}
   \norm{(e^{i\pi(f)}-e^{i\pi(g)})\Psi}&=\norm{e^{-i\pi(g)}e^{i\pi(f)}\Psi-\Psi}=\norm{\int\limits_0^1e^{-i\pi(g)t}\pi(f-g)e^{i\pi(f)t}\Psi dt} \\
                                    &\leq\int\limits_0^1\norm{e^{-i\pi(f)t}\pi(f-g)e^{i\pi(f)t}\Psi}dt \\
                                    &=\int\limits_0^1\norm{(\pi(f-g)+2t\ \Ima\sprod{f-g}{f})\Psi}dt \\
                                    &\leq \norm{\pi(f-g)\Psi}+\abs{\Ima\sprod{g}{f}}\ \norm{\Psi}.
\end{align*}
The lemma now follows from Lemma~\ref{lm:a-like-root-N}.
\end{proof}

\begin{lemma}
\label{lm:FrankSchlein}
Let $f\in L^2(\R^d)$ and let $F: L^2(\R^d)\to L^2(\R^d)\otimes L^2(\R^d)$ be defined by $(F\varphi)(x,k) =  \varphi(x)e^{-ikx}f(k)$. Then, for all $\Psi\in D(H_0)$,
$$
   \norm{a(F)\Psi} \leq C_f \norm{\sqrt{N}(1-\Delta)^{1/2}\Psi}, 
$$
where
$$
   C_f:=\left(\sup\limits_{h\in\R^d}\int\frac{\abs{f(k)}^2}{1+(h-k)^2}dk\right)^{1/2}. 
$$
\end{lemma}

This Lemma is due to Frank and Schlein, see Lemma~10 in~\cite{FrankSchlein2014}. For completeness of the present paper, we give a short proof. It is based on Lemma~\ref{lm:a-like-root-N} with $B=(1-\Delta)^{-1/2}F$.

\begin{proof}
Let $L= (1-\Delta)^{1/2}$ and note that, by \eqref{def:annihilation},
$$
      a(F)\Psi = a(L^{-1}F) L \Psi.
$$
From Lemma~ \ref{lm:a-like-root-N}, it thus follows that 
$$
     \|a(F)\Psi\| \leq \|L^{-1}F\| \|\sqrt{N}L\Psi\|,
$$
which is the desired estimate provided that $ \|L^{-1}F\| \leq C_f$. To prove this, let $\varphi \in L^2(\R^d)$. Then, by definition of $F$ and by Fourier transform, 
\begin{align*}
\|L^{-1}F\varphi\|^2 &= \int\left|(L^{-1}F\varphi)(x,k)\right|^2\,dxdk\\
    &= \int \frac{1}{1+p^2} |\hat\varphi(p+k)|^2|f(k)|^2\,dp\,dk\\
    &= \int \left( \int \frac{1}{1+(p-k)^2} |f(k)|^2\,dk\right) |\hat\varphi(p)|^2\,dp \leq C_f \|\varphi\|^2.
\end{align*}
\end{proof}


\section{An operator core in terms of coherent states}
\label{app:formal}

In this appendix, we apply the formal Expression~\eqref{def-H} to vectors from $D(H)$. By means of formal manipulations, we illustrate the argument given in the Introduction concerning cancellation of "vectors" outside the Hilbert space.

Let $\Omega\in\FF$ denote the vacuum vector. Then, the space
\begin{align*}
   D:=\span\{\gamma\otimes e^{-i\pi(f)}\Omega\ |\ \gamma,f\in C_0^{\infty}(\R^d)\}
\end{align*}
is a core of $H_0$, and hence, $U_{\infty}^*D$ is a core of $H$ by Theorem~\ref{thm:domain-H}. The elements $\Psi\in U_{\infty}^*D$ have the form
\begin{align}
   \Psi(x)=U_{\infty}^*\left(\gamma\otimes e^{-i\pi(f)}\Omega\right)(x)&=\gamma(x)e^{-i\pi(B_{\infty,x}+f)}\Omega \ e^{-i\Ima\sprod{B_{\infty,x}}{f}} \notag \\
   \label{equ:formOfTheVectorsOfTheHereConstructedDenseSubspace}       &=\varphi(x)\eta(x),
\end{align}
where
\begin{align*}
   \varphi(x):=\gamma(x)e^{-i\Ima\sprod{B_{\infty,x}}{f}}e^{-\frac{1}{2}\norm{B_{\infty,x}+f}^2}
\end{align*}
belongs to $C_0^{\infty}(\R^d)$ and
\begin{align*}
   \eta(x):=\sum\limits_{n\geq 0}\frac{1}{n!}a^*(B_{\infty,x}+f)^n\Omega.
\end{align*}
We now formally apply $-\Delta+N+a(G_{\infty})+a^*(G_{\infty})$ to~\eqref{equ:formOfTheVectorsOfTheHereConstructedDenseSubspace}. Using the Leibniz rule to compute $\Delta\Psi$, we obtain
\begin{align*}
   -\Delta\Psi\phantom{(x)}&=(-\Delta\varphi)\eta-2\nabla\varphi\cdot\nabla\eta+\varphi(-\Delta\eta), \\
   N\Psi(x)&=a^*(B_{\infty,x}+f)\Psi(x), \\
   a(G_{\infty,x})\Psi(x)&=\sprod{G_{\infty,x}}{B_{\infty,x}+f}\Psi(x), \\
   a^*(G_{\infty,x})\Psi(x)&=\varphi(x) a^*(G_{\infty,x})\eta(x), \\
   \varphi(x)(-\Delta\eta)(x)&=\varphi(x) a^*(k^2B_{\infty,x})\eta(x)+\varphi(x) a^*(kB_{\infty,x})^2\eta(x).
\end{align*}
All terms on the right-hand side of these five equations are Hilbert space vectors, with the exception of $\varphi a^*(k^2B_{\infty})\eta$ and $\varphi a^*(G_{\infty})\eta$. The sum of these two terms, however, is
$$
   \varphi a^*(k^2B_{\infty}+G_{\infty})\eta=\varphi a^*(G_{\ir}-B_{\infty})\eta,
$$
which is a Hilbert space vector again. Altogether, we get the formal result
\begin{align}
   &(-\Delta+N+a(G_{\infty,x})+a^*(G_{\infty,x}))\Psi(x) \notag \\
   \begin{split}   \label{equ:altogetherFormalResultOfTheFormalComputation}
      &=(-\Delta\varphi)(x)\eta(x)+2(i\nabla \varphi)(x)\cdot a^*(kB_{\infty,x})\eta(x)+a^*(kB_{\infty,x})^2\Psi(x) \\
      &\phantom{=(-\Delta}+a^*(f)\Psi(x)+a^*(G_{\ir,x})\Psi(x)+\sprod{G_{\infty,x}}{B_{\infty,x}+f}\Psi(x),
   \end{split}
\end{align}
which is a Hilbert space vector. A rigorous application of the operator $H=U_{\infty}^*H_{\infty}'U_{\infty}$ (see Theorem~\ref{thm:domain-H}) on the vector $\Psi$ from Equation~\eqref{equ:formOfTheVectorsOfTheHereConstructedDenseSubspace}, which is a long straightforward calculation, leads to exactly the same result~\eqref{equ:altogetherFormalResultOfTheFormalComputation}.


\bigskip\noindent
\textbf{Acknowledgements:} We thank Joachim Kerner and Ioannis
Anapolitanos for many discussions
at an early stage of this work. Ioannis also provided a first version of
the proof of Theorem~\ref{thm:U-exists}. The work of Andreas W\"unsch was
supported by the \emph{Deutsche Forschungsgemeinschaft (DFG)} through the Research
Training Group 1838: \emph{Spectral Theory and Dynamics of Quantum Systems}.



\end{document}